\documentclass[a4paper,11pt]{article}
\usepackage[utf8]{inputenc}
\usepackage{amssymb,amsmath}
\usepackage{enumerate}
\usepackage{theorem}
\usepackage{xcolor,graphicx}
\usepackage{helvet}
\usepackage{multicol}
\usepackage{bm}

\usepackage[left=3.5cm, right=3.5cm, top=2cm, bottom=2cm]{geometry}

\let\ds=\displaystyle

\def\R{\mathbb{R}}
\def\ol{\overline}

\def\<{\langle}
\def\>{\rangle}

\def\v{V}

\def\v{V}

\def\T{\bm{T}}
\def\N{\bm{N}}
\def\B{\bm{B}}
\def\nxi{\bm{\xi}}

\def\U{\bm{U}}
\def\V{\bm{V}}

\def\R{\bm{R}}

\def\D{D}
\def\Partial{\partial}

\def\k{\bm{k}}
\def\ntau{\bm{\tau}}
\def\p{\bm{p}}

\def\f{\bm{f}}
\def\h{\bm{h}}

\def\g{\bm{g}}

\newcommand{\displayskip}[1][1]{%
 \abovedisplayskip=#1\abovedisplayskip
 \belowdisplayskip=#1\belowdisplayskip
 \abovedisplayshortskip=#1\abovedisplayshortskip
 \belowdisplayshortskip=#1\belowdisplayshortskip
 }

\displayskip[.5]

{\theorembodyfont{\itshape}
\newtheorem{theorem}{Theorem}

\newtheorem{definition}[theorem]{Definition}

\newtheorem{lemma}[theorem]{Lemma}

\newtheorem{proposition}[theorem]{Proposition}
\newtheorem{remark}[theorem]{Remark}

}
\newenvironment{proof}{\ignorespaces\par\noindent\textbf{Proof.}\quad}{\hfill$\Box$\par}

\def\X{\mathfrak{X}}

\begin{document}
\title{Integrability aspects of the vortex filament equation for pseudo-null curves}
\author{José del Amor$^{a}$, Ángel Giménez$^{b}$ and Pascual Lucas$^{a}$ \\
$^{a}${\small Departamento de Matemáticas, Universidad de Murcia}\\
{\small Campus de Espinardo, 30100 Murcia, Spain} \\
$^{b}${\small Centro de Investigación Operativa, Universidad Miguel Hernández de Elche}\\
{\small Avda. Universidad s/n, 03202 Elche (Alicante), Spain}\\
}
\date{}

\maketitle

\begin{abstract}
An algebraic background in order to study the integrability properties of pseudo-null curve motions in a $3$-dimensional Lorentzian space form is developed. As an application we delve into the relationship between the Burgers' equation and the pseudo-null vortex filament equation. A recursion operator for the pseudo-null vortex filament equation is also provided.
\end{abstract}

\noindent\textit{Keywords:} Curve evolution; Pseudo-null curves; Lie algebra; Integrability; Burgers' equation.

\noindent\textit{2010 MSC:} 14H70; 17B80; 35Q53; 37K10.

\section{Introduction}
Curves motions in a certain geometric setting in turn induce evolution equations for their invariants which are very often well-known integrable systems. There is a very large body of literature on the relationships between curve motions and integrable systems. It is interesting to stress that having a scalar evolution equation coming from a curve motion helps to elucidate many aspects of its integrability from the intrinsic geometry of the involved curves, and vice versa (see \cite{langer_poisson_1991,beffa_integrable_2002,sanders_integrable_2003, mansfield_evolution_2006,del_amor_hamiltonian_2014,del_amor_lie_2015,gurbuz_moving_2015}). It has been shown in \cite{chou_integrable_2003,chou_motions_2004} that Burgers' hierarchy naturally arises from inextensible motions of plane curves in similarity geometry.  More recently in \cite{grbovic_backlund_2016}, the authors derived the vortex filament equation for pseudo-null curves in the Minkowski space, and proved that the evolution equation for its invariant (pseudo-torsion) is the Burgers' equation. 
Let us remind that Burgers' equation is the nonlinear parabolic equation given by
\begin{equation}\label{eq-burger}
	u_{t}=u_{xx}+u u_{x}, \quad x\in \mathbb{R}, t>0,
\end{equation}
introduced by Burgers' \cite{burgers_mathematical_1948} as a one-dimensional very simplified model of the Navier-Stokes equations. It is often used as a toy model to understand general models in viscous fluid flows and it appears in numerous physical problems.
Burgers' equation is the simplest partial differential equation combining both nonlinear propagation effects and dissipative effects. 
Among many other properties, equation \eqref{eq-burger} possesses the well-known recursion operators
\begin{equation}\label{burger-recursion}
	\begin{aligned}
		\mathcal{R}_{1}&=D_{x}+\frac{1}{2}u+\frac{1}{2}u_{x}D_{x}^{-1}; \\
		\mathcal{R}_{2}&=tD_{x}+\frac{1}{2}(t u+x)+\frac{1}{2}(t u_{x}+1)D_{x}^{-1}. 
	\end{aligned}
\end{equation} 
On the other hand, pseudo-null curve is a singular type of curve which requires to be clarified. In a Lorentzian space there are timelike curves, spacelike curves and null curves according to the causal character of its velocity vector. A pseudo-null curve is a spacelike curve with the constraint that its acceleration vector is null. The mathematical relevance of pseudo-null curves is undoubted, since they arise in a natural way in Lorentzian spaces. New findings about flows of this kind of curves in the present study provide evidences about its physical relevance as well.   

The working outline and the methods used in this paper closely follow the approach adopted in \cite{del_amor_hamiltonian_2014,del_amor_lie_2015,del_amor_null_2016}. Those papers provide significant progress in the search for relationships between geometric motions of curves and integrable systems by means of introducing a suitable Lie bracket on the phase space of the curve evolution. Nevertheless, each one of the different types of curves has its own peculiarities, making it very difficult to present a general treatment covering all the cases at the same time even when you are working in the same ambient space. This paper is intended to adapt this approach to evolutions of pseudo-null curves. The paper is organized as follows. In the next section we briefly highlight some important aspects and tools of integrability. Section 3 concerns the geometry of pseudo-null curves, in particular, we introduce a couple of pseudo-orthonormal frames regarding which the corresponding equations and invariants (curvatures) are defined. 
In Section 4 we investigate pseudo-null curve variations and state the relevant formulas in this setting. 
To search symmetries of pseudo-null curve evolution equations a Lie bracket on a suitable phase space, from a geometrical point of view, is introduced in Section 5. Then we move on to Section 6 in which we apply all elements borrowed from the previous sections to study the integrability aspects of the pseudo-null vortex filament equation. Finally, in Section 7 we give some concluding remarks.


\section{Preliminaries}\label{preliminares}
In this section we summarize some necessary notions and basic definitions from differential calculus which are relevant to the rest of the paper (see \cite{dorfman_dirac_1993,dickey_soliton_2003,blaszak_multi-hamiltonian_2012} for a very complete treatment of the subject).  
Consider $u$ a differentiable function in the real variable $x$, and set
\begin{equation*}
	u^{(m)}=\frac{d^mu}{dx^m},\quad \text{for } m\in \mathbb{N}.
\end{equation*}
Let $\mathcal{P}$ be the algebra of \emph{polynomials} in $u$ and its derivatives of arbitrary order, namely, 
\begin{equation*}
	\mathcal{P}=\mathbb{R}[u^{(m)}:m\in \mathbb{N}].
\end{equation*}
We refer to the elements of $\mathcal{P}$ whose constant term vanishes as $\mathcal{P}_0$. Acting on the algebra $\mathcal{P}$ is defined a derivation $\partial$ obeying
\begin{equation*}
	\begin{cases} 
		\partial(f g)=(\partial f)g+f(\partial g), \\
		\partial(u^{(m)})=u^{(m+1)},
	\end{cases}
\end{equation*}
thereby becoming a differential algebra. 
\begin{remark}
	In a more general setting we can consider $\mathcal{P}$, for example, to be the algebra of \emph{local functions}, i.e. $\mathcal{P}=\bigcup_{j=1}^{\infty}\mathcal{P}_j$, where $\mathcal{P}_j$ is the algebra of locally analytic functions of $u$ and its derivatives up to order $j$ (see \cite{mikhailov_symmetry_1987,sokolov_symmetries_1988,mikhailov_symmetry_1991}). All the results and formulas established in sections 4 and 5 involving the algebra $\mathcal{P}$ remain valid if the differential algebra of polynomials is replaced by the differential algebra of local functions. Nonetheless, the differential algebra of polynomials is sufficient for our purposes.
\end{remark}
It is customary to take $\partial$ as the total derivative $D_x$ which can be viewed as
\begin{equation*}
	D_x=\sum_{m\in \mathbb{N}}u^{(m+1)} \frac{\partial }{\partial u^{(m)}}.
\end{equation*}
In addition to $\partial$, other derivations $\xi$ may also be considered. The action of $\xi$ is determined if we know how $\xi$  acts on the generators of the algebra. Indeed, set $a_{m}=\xi u^{(m)},$ then, for any $f\in \mathcal{P}$ we have
$$\xi f=\sum_{m\in\mathbb{N}} a_{m} \frac{\partial f}{\partial u^{(m)}}.$$
The space of all derivations on $\mathcal{P}$, denoted by $\mathop{\rm der}\nolimits(\mathcal{P})$, is a Lie algebra with respect to the usual commutator
\begin{equation*}
	[\partial_1,\partial_2]=\partial_1\partial_2-\partial_2\partial_1,\quad \partial_1,\partial_2\in \mathop{\rm der}\nolimits(\mathcal{P}).
\end{equation*}
Derivations commuting with the total derivative have important properties. Among others, if $[\xi,\partial]=0$, we have
\begin{equation*}
	a_{m+1}=\xi u^{(m+1)}=\xi\partial u^{(m)}=\partial\xi u^{(m)}=\partial a_{m}.
\end{equation*}
Thus 
\begin{equation*}
	\xi f=\sum_{m\in\mathbb{N}} a^{(m)} \frac{\partial f}{\partial u^{(m)}}.
\end{equation*}
For each element $a$ of $\mathcal{P}$ we shall write
\begin{equation}\label{derivation}
	\partial_a=\sum_{m\in\mathbb{N}} a^{(m)}\frac{\partial }{\partial u^{(m)}}.
\end{equation}
The set of derivations $\partial_a$ is a Lie subalgebra of $\mathop{\rm der}\nolimits(\mathcal{P})$, and it induces a Lie algebra on the space $\mathcal{P}$. Indeed, a direct computation shows that $\partial_a$ are derivations in $\mathcal{P}$ verifying $[\partial_a,\partial_b]=\partial_{[a,b]}$, 
where $[a,b]=\partial_ab-\partial_ba$, i.e.
\begin{equation}\label{scalar-bracket}
	[a,b]=\sum_{m\in\mathbb{N}}\left(a^{(m)}\frac{\partial b}{\partial u^{(m)}}-b^{(m)}\frac{\partial a}{\partial u^{(m)}}\right).
\end{equation}
This latter commutator can also be expressed  with the aid of Fréchet derivatives as $[a,b]=b'[a]-a'[b]$, where
\begin{equation}\label{frechet-derivative}
		a'[b]=\left.\frac{d}{dt}\right|_{t=0}\partial_{a}(u+tb).
\end{equation}
We will refer to $\partial_{a}$ as an evolution derivation (or a vector field, provided that no confusion is possible), and the algebra of all evolution derivations will be denoted by  $\mathop{\rm der}\nolimits^*(\mathcal{P})$. Observe that, in particular, if we take $a=u'$ then $\partial=\partial_{u'}$, being $u'=D_x u$.

Consider an evolution equation of the form
\begin{equation}\label{evolution-equation}
	\frac{\partial u}{\partial t}=f(u,u^{(1)},u^{(2)},\ldots)
\end{equation}
where $f$ is an element of $\mathcal{P}$. An element $\sigma\in \mathcal{P}$ is called a symmetry of the evolution equation \eqref{evolution-equation} if and only if  $[f,\sigma]=f'[\sigma]-\sigma'[f]=0$.
Symmetries of integrable equations can often be generated by recursion operators which are linear operators mapping a symmetry to a new symmetry. A linear differential operator $\mathcal{R}:\mathcal{P}\rightarrow \mathcal{P}$ is a recursion operator for the evolution equation \eqref{evolution-equation} if it is invariant under $f$, i.e., $L_f \mathcal{R}=0$, where $L_f$ is the Lie derivative acting as $L_{f}a=[f,a]$ for all $a\in \mathcal{P}$. $\mathcal{R}$ is said to be hereditary if for an arbitrary $f\in \mathcal{P}$ the relation $L_{\mathcal{R}f}\mathcal{R}=\mathcal{R}L_f \mathcal{R}$ is verified.

\section{Geometry of pseudo-null curves}
A Lorentzian manifold $(M^n_1,g)$ is an $n$-dimensional differentiable manifold $M^n_1$ endowed with a non-degenerate metric tensor $g$ with signature $(n-1,1)$. The metric tensor $g$ will be also denoted by $\left\langle\cdot , \cdot\right\rangle$ and the Levi-Civita connection by $\nabla$. The sectional curvature of a non-degenerate plane generated by $\left\{u,v\right\}$ is $$K(u,v)=\frac{\left\langle R(u,v)u,v \right\rangle}{\left\langle u,u \right\rangle \left\langle v,v \right\rangle-\left\langle u,v \right\rangle^2},$$
where $R$ is the Lorentzian curvature tensor given by
\begin{equation}\label{eq-curvature-tensor}
	R(X,Y)Z=-\nabla_X\nabla_YZ+\nabla_Y\nabla_XZ+\nabla_{[X,Y]}Z.
\end{equation}
Lorentzian manifolds with constant sectional curvature are called Lorenzian space forms. It is a well-known fact that the curvature tensor $R$ adopts a simple formula in these manifolds:
\begin{equation}\label{eq-constant-curvature}
	R(X,Y)Z=G \left\{\left\langle Z,X \right\rangle Y- \left\langle Z,Y \right\rangle X\right\},
\end{equation}
where $G$ is the constant sectional curvature. 
A $3$-dimensional Lorentzian space form  $M^3_1(G)$ can be realized as $\mathbb{R}^3_1$ (when $G=0$) and as a hyperquadric in $\mathbb{R}^4_1$ or $\mathbb{R}^4_2$ (when $G>0$ or $G<0$, respectively).  
Consider $\gamma:I\to M^3_1(G)$, $\gamma\equiv\gamma(s)$ an arc-length parametrized spacelike curve such that $\nabla_{T}T$ is a null vector, where $T(s)=\gamma'(s)$, and $\nabla_{T}$ stands for the covariant derivative along $\gamma$. The Frenet equations for $\gamma$ are
\begin{equation}\label{eq-Frenet}
	\begin{aligned}
		\nabla_TT&=N,\\
		\nabla_TN&=\tau N,\\
		\nabla_TB&=T-\tau B,
	\end{aligned}
\end{equation}
where $\tau$ will be referred as the \textit{pseudo-torsion} of $\gamma$. The frame $\{T,N,B\}$ will be called the \textit{Frenet frame}, and it is a pseudo-orthonormal frame verifying:
\begin{align*}
\<T,T\>&=1, \<T,N\>=\<T,B\>=0,\\
\<N,B\>&=-1, \<N,N\>=\<B,B\>=0.
\end{align*}
We refer to this type of curves as \emph{pseudo-null curves} (see \cite{walrave,rafael_lopez}). Besides, if $\tau: I \rightarrow \mathbb{R}$ is a smooth function, then there exists an unique pseudo-null curve (up to an action of an isometry in $M^3_1(G)$) with pseudo-torsion $\tau$. 	The Gauss equation for $M^3_1(G)$ establishes
\begin{align*}
X' &= \nabla_TX-G\<T,X\>\gamma,
\end{align*}
where $X'\equiv X'(s)$ stands for the usual derivative in $ \mathbb{R}^3_1, \mathbb{R}^4_1$ or $ \mathbb{R}^4_2$ according to the sign of the sectional curvature $G$.

Let $\gamma$ be a pseudo-null curve in $M^3_1(G)$ with Frenet frame $\left\{T,N,B\right\}$. We fix a constant $c>0$ and consider the vector field along $\gamma$ given by  
\begin{equation*}
	\xi(s)=\lambda(s)N(s), \quad \lambda(s_0)=\frac{1}{c}.
\end{equation*}
Derivating $\xi(s)$ with respect to the parameter $s$ we obtain:
\begin{equation}
	\begin{aligned}
		\xi'(s)&=\lambda'(s)N(s)+\lambda(s)N'(s)\\ 
		&=\lambda'(s)N(s)+\lambda(s)\left(\nabla_{T(s)}N(s)-G \left\langle T,N \right\rangle\gamma(s)\right)\\
		&=(\lambda'(s)+\lambda(s)\tau(s))N(s).
	\end{aligned}
\end{equation}
It follows that $\xi$ is a parallel vector field along of $\gamma$ (i.e. $\xi'(s)=0$) if and only if $\lambda(s)$ verifies the differential equation $\lambda'(s)+\lambda(s)\tau(s)=0$. For $s_0\in I$, the solution of this equation with initial condition $\lambda(s_0)=\frac{1}{c}$ is given by
$$\lambda(s)=\frac{1}{c} e^{-\int_{s_0}^{s} \tau(p)\;dp}.$$
In short, we can construct a parallel null vector field $\xi_0$ along $\gamma$ orthogonal to $\gamma'=T$, and such that the pseudo-null curve $\gamma$ is contained in the null hyperplane orthogonal to $\xi_0$. More precisely, the curve $\gamma$ is contained in the lightlike surface 
\begin{equation*}
	S= \left\{ x\in M^3_1(G): \left\langle x,\xi_0 \right\rangle = r \right\},		
\end{equation*}
where $r$ is a constant. In fact, the lightlike surface $S$ can be locally parametrized by
$$X(s,t)=\gamma(s)+t \xi_0,$$
and thus $S$ is a lightlike cylinder whose rulings are null geodesic lines in the direction of $\xi_0$.
Accordingly we state the following result.
\begin{proposition}
	Let $\gamma:I \rightarrow M^3_1(G)$ be an arc-length parametrized spacelike curve. Then $\gamma$ is a pseudo-null curve if and only if $\gamma$ is congruent to a non-geodesic spacelike curve in a lightlike cylindrical surface.  
\end{proposition}
Some results about null ruled surfaces in $3$ and $4$-dimensional Minkowski spaces can be found in \cite{kilic_ligthlike_2006,ozturk_international_2016}.  
\begin{remark}\label{remark-string}
A null string is defined as a lightlike surface whose generators are geodesics (see \cite{schild_classical_1977,perlick_totally_2005}). Consequently, lightlike cylindrical surfaces are particular cases of null strings in which the generators have a fixed direction given by the null ruling.  
\end{remark}

It is clear that the normal vector field $N(s)$ is always in the null fixed direction orthogonal to the hyperplane in which the curve is contained. Going somewhat further, if we take
\begin{equation}\label{eq-relation1}
	k(s)=\frac{1}{\lambda(s)}=c e^{\int_{s_0}^{s} \tau(p)\;dp}\quad \Leftrightarrow\quad k'(s)=\tau(s)k(s), \quad k(s_0)=c,
\end{equation}
and
\begin{equation}\label{eq-relation2}
	\xi(s)= \frac{1}{k(s)}N(s), \qquad \eta(s)= k(s)B(s),
\end{equation}
then it is verified that
\begin{equation}\label{parallel-frame}
	\begin{aligned}
		\nabla_{T}T&=k \xi, \\ \nabla_{T}\xi&=0, \\ \nabla_{T}\eta &= k T.
	\end{aligned}
\end{equation}
The pseudo-orthonormal frame $\left\{T,\xi,\eta\right\}$ will be called a \textit{parallel frame} of $\gamma$ and $k$ its \textit{pseudo-curvature}.
Observe that unlike the Frenet frame, the parallel frame is not unique and depends on a constant which merely determines the null directions scaling factor of the frame. In the next section the calculations will be developed attending to the Frenet frame and its pseudo-torsion, whilst it could be equally performed on the basis of a parallel frame.

\section{Pseudo-null curve variations}
Let $\gamma$ be a pseudo-null curve, for the sake of simplicity the letter $\gamma$ will also denote a variation of pseudo-null curves (pseudo-null curve variation) $\gamma=\gamma(p,t):I\times(-\varepsilon,\varepsilon)\rightarrow M^3_1(G)$ with $\gamma(p,0)$ the initial pseudo-null curve. Associated with such a variation is the variation vector field $\v(p)=\v(p,0)$, where $\v=\v(p,t)=\frac{\partial \gamma}{\partial t}(p,t)$. We denote by $v$ the velocity function verifying $\ol T(p,t)=v(p,t)T(p,t)$, where $\ol T(p,t)=\frac{\partial \gamma}{\partial p}(p,t)$, and by $\frac{D}{\partial t}$ the covariant derivative along the curves $\gamma_p(t)=\gamma(p,t)$. We will also write $\gamma(s,t), \tau(s,t), \v(s,t),$ etc., for the corresponding objects in the arc-length parametrization.

\begin{lemma}\label{lema-variation1}
	If $\gamma$ is a pseudo-null curve variation, then 
	\begin{enumerate}[(a)]\itemsep0pt
		\item $\left.\frac{\partial v}{\partial t}\right|_{t=0} = \rho_V v, \qquad \rho_V=\left\langle \nabla_T \v,T  \right\rangle$, 
		\item $[V,T] = -\rho_V T$,
		\item $\left\langle \nabla_T^{2} \v+G\v, N \right\rangle=0$,	 
	\end{enumerate}
	where $v$ is the velocity and $\v$ its associated variation vector field.
\end{lemma}
\begin{proof}
	The velocity $v\equiv v(p,t)$ verifies $v^{2}=\left\langle \ol T, \ol T \right\rangle$ that, together with the property $[V,\ol T]=0$ leads to
	\begin{equation*}
		2 v V(v)= 2 \left\langle \nabla_{V}\ol T, \ol T \right\rangle=2 \left\langle \nabla_{\ol T} V, \ol T \right\rangle = 2 v^{2} \left\langle \nabla_{T} V, T \right\rangle.
	\end{equation*}
	As a consequence it follows (a) and (b). Moreover, 
	$$\<\nabla_{\ol T}{\ol T},\nabla_{\ol T}{\ol T}\>=v^2T(v)^2.$$
	To prove (c) we first take the derivative in the $V$ direction on the right side term of the latter equation  and use (a) and (b) to get
	\[
	V(v^2T(v)^2)=2v^2\rho_VT(v)^2+2v^3T(v)T(\rho_V).
	\]
	By doing the same on the left side term we obtain
	\[
	V(\<\nabla_{\ol T}{\ol T},\nabla_{\ol T}{\ol T}\>)=2v^4\<\nabla_T^2V,N\>+2Gv^4\<V,N\>+2v^3T(v)T(\rho_V) +2v^2T(v)^2\rho_V
	\]
	From equating and simplifying both expressions it follows (c).
\end{proof}
It is therefore justified to introduce the following definition. 
\begin{definition}
	Let $\mathfrak{X}(\gamma)$ be the set of smooth vector fields along $\gamma$. We say that $\v\in\mathfrak{X}(\gamma)$  locally preserves the pseudo-null character if $\left\langle \nabla_T^{2} \v+G \v, N \right\rangle=0$. We also say that $\v$ locally preserves the arc-length parameter along $\gamma$ if $\rho_{\v}=\left\langle \nabla_{T}\v,T \right\rangle=0$.
\end{definition}
Let $\v=f_{\v} T+g_{\v} N + h_{\v} B$ be a generic vector field along a pseudo-null curve $\gamma$. From the Frenet equations we easily calculate the first and the second covariant derivative of $\v$,
\begin{equation*}
	\begin{aligned}
		\nabla_{T}\v &=  \left(f'_{\v}+h_{\v}\right)T + \left(f_{\v}+g'_{\v}+ \tau g_{\v} \right) N + \left(h'_{\v}- \tau h_{\v} \right) B, \\
		\nabla_{T}^{2}\v &=  \left(f''_{\v}+2 h'_{\v}- \tau h_{\v} \right) T+ \left(2 f'_{\v}+ \tau f_{\v} +g''_{\v}+2 \tau  g'_{\v} + \left(\tau '+\tau ^2\right) g_{\v} +h_{\v}\right) N \\ 
		&\quad +\left(h''_{\v}-2 \tau  h'_{\v}+(\tau ^2-\tau') h_{\v}\right) B.
	\end{aligned}
\end{equation*}
Note that if $\v$ locally preserves the pseudo-null character then the binormal component $h_{\v}$ is limited to be a solution of the second order homogeneous differential equation given by
\begin{equation}\label{eq-comp1}
	h_{\v}''-2 \tau  h_{\v}'+(\tau^2-\tau'+G) h_{\v}=0,
\end{equation}
whose coefficients are completely determined by the pseudo-torsion of the curve. If, moreover, $\v$ locally preserves the arc-length parameter then $f$ is calculated from $h$ by means of the equation     
\begin{equation}\label{eq-comp2}
	f'_{\v}+h_{\v}=0.
\end{equation}

To complete the study of variational behavior of pseudo-null curves we still have to calculate the Frenet frame variation and the pseudo-torsion variation. 
\begin{lemma}\label{lema-variation2}
	Let $\gamma_{t}(p)$ be a pseudo-null variation starting from a pseudo-null curve $\gamma_{0}(p)$. The following assertions hold: 
	\begin{enumerate}[(a)]\itemsep0pt
		\item $\ds\dfrac{DT}{\partial t}\bigg|_{t = 0}  = \varphi_V N + \psi_V B$;
		\item $\ds\dfrac{DN}{\partial t}\bigg|_{t = 0}  = \psi_V T + \alpha_{V} N $;
		\item $\ds\dfrac{DB}{\partial t}\bigg|_{t = 0}  =  \varphi_V T -\alpha_{V} B$;
		\item $\ds\dfrac{\partial \tau}{\partial t}\bigg|_{t = 0}  = \alpha'_V +\psi_V-\tau\rho_V$;
	\end{enumerate}
	where  
	\begin{equation}\label{eq-notacion1}
		\begin{split}
			\varphi_V &= -\langle\nabla_T V,B\rangle=f_V+g'_V+\tau g_V; \\ 
			\psi_V &=-\langle\nabla_T V, N\rangle=h'_V-\tau h_V; \\
			\alpha_V &= \varphi'_V +\tau\varphi_V + Gg_V -\rho_V.
		\end{split}	
	\end{equation}
\end{lemma}
\begin{proof}
	Throughout the proof the equations \eqref{eq-constant-curvature} and \eqref{eq-Frenet}, as well as Lemma \ref{lema-variation1} are  repeatedly applied. Denoting also by $\nabla_{V}=\frac{D}{dt}$ the covariant derivative, we first have that
			\begin{align*}
				\nabla_V T &= \langle\nabla_V T,T\rangle T -\langle\nabla_V T,B\rangle N-\langle\nabla_V T,N\rangle B\\
						   &= - \langle\nabla_T V-\rho_V T,B\rangle N - \langle\nabla_T V-\rho_V T,N\rangle B\\
						   &= \varphi_V N + \psi_V B.  
			\end{align*}
	We also find that
			\begin{align*}
				\nabla_V N &= \nabla_V\nabla_T T = \nabla_T\nabla_V T + \nabla_{[V,T]}T - G\{\langle T,V\rangle T - \langle T,T\rangle V\}\\
						   &= \nabla_T (\varphi_V N + \psi_V B)-\rho_V\nabla_{T}T-G\langle T,V\rangle T + GV \\
						   &=\psi_V T + \left(\varphi'_V +\tau\varphi_V+G g_{V}-\rho_V\right)N
			\end{align*}
	To obtain (c) we proceed as follows:
	\begin{equation*}
		\begin{aligned}
			\nabla_V B  &=\langle\nabla_V B,T\rangle T -\langle\nabla_V B,B\rangle N -\langle\nabla_V B,N\rangle B \\
				&=-\langle B,\nabla_VT\rangle T +\langle B,\nabla_V N\rangle B \\
				&=\varphi_V T - \left(\varphi'_V +\tau\varphi_V+G g_{V}-\rho_V\right) B.
		\end{aligned}
	\end{equation*}	
	Taking the covariant derivative $\nabla_V$ to the second Frenet equation we get
	\begin{equation}\label{e3}
		\begin{aligned}
			\nabla_V \nabla_T N &= \nabla_V (\tau N) = V(\tau)N + \tau\nabla_V N.  
		\end{aligned}
	\end{equation}
	The last equation leads to
	\begin{equation}\label{e5}
		\begin{aligned}
			V(\tau) &= \left\langle  -\nabla_V\nabla_T N +\tau\nabla_V N , B\right\rangle \\
			&= \left\langle  -\nabla_T\nabla_V N +\tau\rho_V N +G\langle N , V \rangle T +\tau\nabla_V N , B\right\rangle \\
			&=\psi_{V}+\left(\varphi'_V +\tau\varphi_V+G g_{V}-\rho_V\right)'-\tau \rho_{V},
		\end{aligned}
	\end{equation}
	which establishes property (d). 
\end{proof}

Consider $\Lambda$ the set of arc-length parametrized pseudo-null curves in $M^3_1(G)$. Given a curve $\gamma\in\Lambda$, the tangent space $T_\gamma\Lambda$ is the set of variation vector fields $V\in\X(\gamma)$ associated to variations of curves in $\Lambda$ starting from $\gamma$.
\begin{proposition}\label{prop-tangent1}
	A vector field $\v$ along $\gamma\in\Lambda$ is tangent to $\Lambda$ if and only if it locally preserves the pseudo-null character and the arc-length parameter, thus
	\begin{equation*}
		T_\gamma\Lambda=\{V\in\X(\gamma)\,:\,\<\nabla_T^2V+GV,N\>=\<\nabla_TV,T\>=0\}.
	\end{equation*}
	Therefore, if we take
	\begin{align}
	V &= f_{\v}T+g_{\v}N+h_{\v}B,\label{eq4}
	\end{align}
	where $f_{\v},g_{\v},h_{\v}$ are smooth functions, then $V\in T_\gamma\Lambda$ if and only if
	\begin{equation}\label{eq4+}
	\left[(D_s-\tau)^2+G\right](h_{\v})=0\quad\text{and}\quad f'_{\v}+h_{\v}=0.
	\end{equation}
	In such a case 
	\begin{equation}\label{eq5}
		\begin{aligned}
			V(\tau)&=\varphi''_{\v}+(\tau\varphi_{\v})'+\psi_{\v}+G g'_{\v} \\
			&=  g'''_{\v}+2 \tau  g''_{\v} + \left(3 \tau' + \tau^2 + G \right) g'_{\v} + \left(\tau''+2 \tau  \tau'  \right) g_{\v} + 2 \tau  f'_{\v}+ \tau' f_{\v}
		\end{aligned}	
	\end{equation}
	where $\varphi_{\v}$ and $\psi_{\v}$ are given by \eqref{eq-notacion1}.
\end{proposition} 
\begin{proof}
	If $\v\in T_\gamma\Lambda$, Lemma \eqref{lema-variation1} implies 
	\begin{equation}\label{eq-compatibilidad}
			\<\nabla_T^2V+GV,N\>=\<\nabla_TV,T\>=0.
	\end{equation}
	Equations \eqref{eq4+} and \eqref{eq5} are a consequence of \eqref{eq-comp1}, \eqref{eq-comp2} and Lemma \ref{lema-variation2}(d), respectively.
	Conversely, let $\v=f_{\v} T+g_{\v} N+ h_{\v} B$ be a vector field satisfying \eqref{eq-compatibilidad}, and 
	consider the matrices
	\begin{equation}
	K_{\gamma}=\left(\begin{array}{rrrr}
	0 & -G & 0 & 0\\
	1 & 0 & 0 & 1\\
	0 & 1 & \tau_{\gamma} & 0\\
	0 & 0 & 0 & -\tau_{\gamma}
	\end{array}\right)\qquad\text{and}\qquad
	P= \left(\begin{array}{rrrr}
	0 & -Gf_{\v} & Gh_{\v} & Gg_{\v}\\
	f_{\v} & 0 & \psi_{\v} & \varphi_{\v}\\
	g_{\v} & \varphi_{\v} & \alpha_{\v} & 0\\
	h_{\v} & \psi_{\v} & 0 & -\alpha_{\v}
	\end{array}\right),
	\end{equation}
	where
	\begin{align}
	\varphi_{\v} &= g_{\v}'+\tau g_{\v}+f_{\v},\label{neq2}\\
	\psi_{\v} &= h_{\v}'-\tau h_{\v},\label{neq3}\\
	\alpha_{\v} &= \varphi_{\v}'+\tau\varphi_{\v}+Gg_{\v}.\label{neq5}
	\end{align}
	By using equations \eqref{eq4+} we obtain
	\begin{equation}
	\frac{\partial P}{\partial s}=C-[K_\gamma,P],
	\end{equation}
	where 
	\begin{equation}
	C=\left(\begin{array}{rrrr}
	0 & 0 & 0 & 0\\
	0 & 0 & 0 & 0\\
	0 & 0 & c & 0\\
	0 & 0 & 0 & -c
	\end{array}\right),\qquad c=\alpha_{\v}'+\psi_{\v}.
	\end{equation}
	The same reasoning as in Lemma 1 of \cite{musso_hamiltonian_2010} proves that there exists a pseudo-null curve variation $\gamma(s,t)$ of $\gamma(s)$ whose variation vector field is $\v$.
\end{proof}

As we mentioned above, a parallel frame satisfying the equations \eqref{parallel-frame} could also be used to achieve the corresponding variational results in a similar way. However, for the purpose of determining the pseudo-curvature variation $V(k)$, we shall use the formula of $V(\tau)$ in Lemma \ref{lema-variation2}(d) and the relationships \eqref{eq-relation1} and \eqref{eq-relation2}.
Considering a vector field $V=\tilde{f}T+\tilde{g}\xi+\tilde{h}\eta$, we easily get 
\begin{equation}\label{eq-relations}
	\begin{aligned}
	f&=\tilde{f}; \quad g= \frac{1}{k}\tilde{g}; \quad h=k\tilde{h}; \quad \rho=\tilde{\rho}; \quad \varphi= \frac{1}{k}\tilde{\varphi}; \quad \psi=k\tilde{\psi}; \\
	\tilde{\rho}&=\tilde{f}'+k\tilde{h};\quad \tilde{\varphi}=k\tilde{f}+\tilde{g}'; \quad \tilde{\psi}=\tilde{h}',
	\end{aligned}
\end{equation}
where $\tilde{\varphi}=-\<\nabla_TV,\eta\>$, $\tilde{\psi}=-\<\nabla_TV,\xi\>$ and $\tilde{\rho}=\<\nabla_TV,T\>$. Derivating the equation $\tau k=k'$ in the direction of $V$ we obtain
$$V(\tau)k+\tau V(k)=TV(k)-\rho k',$$
where we have used Lemma \ref{lema-variation1}(c) in the right-hand side. Multiplying both sides by $k$ and rearranging we obtain
$$(V(k))'k-k'V(k)=V(\tau)k^2+\rho k k'.$$
We now divide by $k^2$ and replace $V(\tau)$ by the expression in \eqref{e5} to achieve
$$ \left( \frac{V(k)}{k} \right)'=\varphi''+(\tau \varphi)'+Gg'+\psi-\rho'.$$
Finally, the relations given in \eqref{eq-relations} lead to the formula
\begin{equation}
	\begin{aligned}
		V(k)&=\tilde{\varphi}'+G\tilde{g}+k D_s^{-1}(k\tilde{\psi})-k\tilde{\rho} \\
		&=\tilde{g}''+k'\tilde{f}+G\tilde{g}-k D_s^{-1}(k'\tilde{h}).
	\end{aligned}
\end{equation}
As a consequence we state the corresponding version of Proposition \ref{prop-tangent1} in terms of the pseudo-curvature $k$.
\begin{proposition}\label{prop-tangent2}
	A vector field $\v$ along $\gamma\in\Lambda$ is tangent to $\Lambda$ if and only if it locally preserves the pseudo-null character and the arc-length parameter, thus
	\begin{equation*}
		T_\gamma\Lambda=\{V\in\X(\gamma)\,:\,\<\nabla_T^2V+GV,\xi\>=\<\nabla_TV,T\>=0\}.
	\end{equation*}
	Therefore, if we take
	\begin{align}
	V &= \tilde{f}_{\v} T+\tilde{g}_{\v}\xi+\tilde{h}_{\v}\eta,
	\end{align}
	where $\tilde{f}_{\v},\tilde{g}_{\v},\tilde{h}_{\v}$ are smooth functions, then $V\in T_\gamma\Lambda$ if and only if
	\begin{equation}
		\tilde{h}_{\v}''+ G \tilde{h}_{\v}=0\quad\text{and}\quad \tilde{f}_{\v}'+k \tilde{h}_{\v}=0.
	\end{equation}
	In such a case 
	\begin{equation}\label{eq-curvature-evolution}
		\begin{aligned}
			V(k)&=\tilde{\varphi}_{\v}'+ G \tilde{g}_{\v} + k \partial_{s}^{-1}(k \tilde{\psi}_{\v})+d_1 k\\
			&= \tilde{g}''+k'\tilde{f}+G\tilde{g}-k \partial_s^{-1}(k'\tilde{h})+d_2 k,
		\end{aligned}	
	\end{equation}
	where  the functions $\tilde{\varphi}$ and $\tilde{\psi}$ are given by \eqref{eq-relations}, $d_1$ and $d_2$ are constants of integration, and $\partial_s^{-1}$ is the anti-derivative operator verifying that $\partial_s^{-1}\circ\partial_s=I$ when acting on $\mathcal{P}_0$.
\end{proposition}

	To close this section we shall find particular vector fields generating interesting flows.
	Let us consider the specific vector field $V=aT+bN$, where $a,b$ are constants, $b\neq0$. The equation \eqref{eq5} applied to this vector field gives rise to
	\begin{equation}\label{eq12}
	\tau_t=b\tau_{ss}+2b\tau\tau_s+a\tau_s.
	\end{equation}
	As a matter of fact, the latter equation constitutes the Burgers' equation \eqref{eq-burger} through the transformation
	$$\tau(s,t)=\frac1{2\sqrt{b}}u\big(b^{-1/2}s,t\big)-\frac a{2b}.$$

	In particular, if we take $a=0$ and $b=1$, we get the pseudo-null vortex filament equation (see \cite{grbovic_backlund_2016}) given by
	\begin{equation*}
		\gamma_{t}=N=\gamma_{s}\times\gamma_{ss},
	\end{equation*}
	whose pseudo-torsion flow is
	\begin{equation}\label{eq-mburger}
		\tau_t=\tau_{ss}+2\tau\tau_s.
	\end{equation}
	The corresponding recursion operator for the equation \eqref{eq-mburger} (after rescaling $\mathcal{R}_{1}$ in \eqref{burger-recursion}) becomes 
	\begin{equation}\label{mburger-recursion}
		\mathcal{R}=\tau_{s}D_{s}^{-1}+\tau+D_{s}.
	\end{equation}

We can generalize somewhat further this situation. Let us consider a pseudo-null variation $\gamma$ preserving the arc-length parameter such that all the curves $\gamma_{t}$ of the variation are contained in the same lightlike cylindrical surface. This means that the variation vector field $V=f_{\v} T+g_{\v} N+h_{\v} B$ associated to such a variation verifies that $h_{\v}=0$ and, from \eqref{eq-comp2}, $f_{\v}$ is constant. The tangential part of $V$ does not modify the geometry of the curve variation, just its parametrization. In this case, since the tangential coefficient is constant, this is nothing but a linear re-parametrization, for that reason it is not restrictive to take $f_{\v}=0$ and so $V=g_{\v} N$. According to the equation \eqref{eq5} we calculate the evolution of $\tau$ as follows:
\begin{equation}\label{eq-Vtau}
	\begin{aligned}
		V(\tau)&=\varphi''_{\v}+(\tau \varphi_{\v})'+G g'_{\v} \\
			&=\mathcal{R}(\varphi'_{\v})+G g'_{\v} \\
			&=\mathcal{R}(g''_{\v}+(\tau g_{\v})')+G g'_{\v} \\
			&=(\mathcal{R}^{2}+G)(g'_{\v}).
	\end{aligned}
\end{equation}
The latter equation gives the pseudo-torsion variation in terms of the recursion operator of the Burgers' equation. 
\begin{remark}\label{remark-recursion}
	It is easy to show that since $\mathcal{R}$ is a recursion operator for the Burgers' equation, then so is the operator $\mathcal{R}^{2}+G$. This fact shall be employed in Section \ref{sect-hierarchies} to construct a recursion operator for the pseudo-null vortex filament equation.
\end{remark}

\section{A Lie algebra structure on local vector fields}\label{sect-lie-algebra}
This section is intended for defining a Lie algebra structure on the set of local vector fields which locally preserve the pseudo-null character. This Lie algebra will be the cornerstone to introduce an appropriate  formal variational calculus on the phase space of pseudo-null curves flows in Section \ref{sect-hierarchies}, thereby enabling the definition of concepts such as symmetry,  recursion operator, and so on. To this end, we need first to set up the spaces in which we are going to work. Let $\tau$ be a smooth function defined on an interval $I$ and set $\mathcal{P}$ the real algebra of polynomials in  $\tau$ and its derivatives of arbitrary order, i.e.,
$$\mathcal{P}=\mathbb{R}\left[ \tau^{(n)} : n\in\mathbb{N}\right].$$ 

Let $\gamma : I\rightarrow M^3_1(G)$ be a pseudo-null curve with pseudo-torsion $\tau$, and consider the set of vector fields along $\gamma$ whose components are polynomial functions
\begin{equation*}
	\mathfrak{X}_{\mathcal{P}}(\gamma)=\{V=f_V T +  g_V N + h_V B \in \mathfrak{X}(\gamma) : f_V , g_V, h_V \in\mathcal{P}\}.
\end{equation*}
An element of $\mathfrak{X}_{\mathcal{P}}(\gamma)$ will be called a $\mathcal{P}$-local vector field along $\gamma$. The set of $\mathcal{P}$-local vector fields (locally preserving the pseudo-null character) will be denoted by
\begin{equation*}
	\mathfrak{X}_{\mathcal{P}}^{*}(\gamma)=\{V=f_V T + g_V N + h_V B \in \mathfrak{X}_{\mathcal{P}}(\gamma) : h_{\v}''-2 \tau  h_{\v}'+(\tau^2-\tau'+G) h_{\v}=0\},
\end{equation*}
and within it, the $\mathcal{P}$-local variation vector fields locally preserving arc-length parameter are described as
\begin{equation}\label{tangent-gamma}
	T_{\mathcal{P},\gamma}\Lambda=T_{\gamma}(\Lambda)\cap \mathfrak{X}_{\mathcal{P}}^{*}(\gamma)=\{V\in \mathfrak{X}_{\mathcal{P}}^{*}(\gamma) : f'_{V}=-h_{V} \}.
\end{equation}
We now introduce a convenient derivation on both the differential algebra and the local vector fields along a pseudo-null curve. Following \cite{del_amor_hamiltonian_2014,del_amor_null_2016} and bearing in mind Lemma \ref{lema-variation2}, given $\v\in \mathfrak{X}^*_{\mathcal{P}}(\gamma)$, we denote by $D_\v:\mathfrak{X}_{\mathcal{P}}(\gamma)\rightarrow \mathfrak{X}_{\mathcal{P}}(\gamma)$ the unique tensor derivation fulfilling:
\begin{enumerate}[(1)]
			\item $\displaystyle V(\tau) =  \varphi''_V +(\tau\varphi_V)' +\psi_V -\rho'_V -\tau\rho_V+Gg'_V=\alpha'_V +\psi_V-\tau\rho_V$;
			\item  $\displaystyle V(f')   =  V(f)' - \rho_V f' \text{ for all } f\in\mathcal{P}$;
			\item  $\displaystyle \left(
					\begin{array}{c}
							\nabla_V T    \\
							\nabla_V N    \\
							\nabla_V B    \\
					\end{array}
				\right) = 	\left(
								\begin{array}{ccc}
									0          & \varphi_V   &\psi_V    \\
									\psi_V     & \alpha_V     &0   \\
									\varphi_V  &     0       &-\alpha_V  \\
								\end{array}
							\right)
										 \cdot \left(
													\begin{array}{c}
														 T    \\
														 N    \\
														 B    \\
													\end{array}
												\right)$,
\end{enumerate}
where $\varphi_V, \psi_V$ and $\alpha_V$ are given by \eqref{eq-notacion1}. We now restrict our definition of Lie bracket only on the set $\mathfrak{X}_{\mathcal{P}}^{*}(\gamma)$, which will be enough for our purposes.

\begin{theorem}\label{prop-lie-bracket}
	Let $\gamma$ be a pseudo-null curve in $\Lambda$ and consider $[\cdot,\cdot]_{\gamma}:\mathfrak{X}^*_{\mathcal{P}}(\gamma)\times \mathfrak{X}^*_{\mathcal{P}}(\gamma)\rightarrow \mathfrak{X}^*_{\mathcal{P}}(\gamma)$ the map given by
	$$[\v_{1},\v_{2}]_{\gamma}=D_{\v_{1}}\v_{2}-D_{\v_{2}}\v_{1}.$$  Then:
	\begin{enumerate}[(a)]
		\item $[\cdot,\cdot]_\gamma$ is well defined, i.e., if $\v_{1},\v_{2}\in\mathfrak{X}^*_{\mathcal{P}}(\gamma)$ then $[\v_{1},\v_{2}]_\gamma\in \mathfrak{X}^*_{\mathcal{P}}(\gamma)$.
		\item $[\v_{1},\v_{2}]_{\gamma}(f)=\v_{1}\v_{2}(f)-\v_{2}\v_{1}(f)$ for all $f\in \mathcal{P}$.
		\item $[\cdot,\cdot]_{\gamma}$ is skew-symmetric.
		\item For $V_1,V_2\in \mathfrak{X}^*_{\mathcal{P}}(\gamma)$ and $U\in \mathfrak{X}_{\mathcal{P}}(\gamma)$ we have 
		\begin{equation}\label{formula-curvatura}
			D_{[V_1,V_2]_\gamma}U-D_{V_1}D_{V_2}U+D_{V_2}D_{V_1}U=G \left(\left\langle U,V_{1} \right\rangle V_{2}-\left\langle U,V_{2} \right\rangle V_{1}\right).
		\end{equation}
		\item $[\cdot,\cdot]_{\gamma}$ satisfies the Jacobi identity.
		\item $[\cdot,\cdot]_{\gamma}$ is closed for elements in $T_{\mathcal{P},\gamma}(\Lambda)$, i.e., if $V,U \in T_{\mathcal{P},\gamma}(\Lambda)$, then $[V,U]_{\gamma}\in T_{\mathcal{P},\gamma}(\Lambda)$.
	\end{enumerate}
\end{theorem}
\begin{proof}
To facilitate the writing of the proof we introduce a couple of operators.
The operator $J$ is defined as $J(x,y)=(y,-x)$, where $(x,y)$ is either a pair of functions or a pair of vector fields. Let us denote by $V_i = f_i T +g_i N + h_i B$, $i=1,2$, vector fields in $\mathfrak{X}^*_{\mathcal{P}}(\gamma)$, and write a pair of elements associated to the vector fields $V_1$ and $V_2$ with a bar over it, for instance, $\bar{V} = (V_1,V_2)$, $\bar{f}=(f_1,f_2)$, $\bar{g}=(g_1,g_2)$, $\bar{h}=(h_1,h_2)$, 
$\bar{\varphi}=(\varphi_1,\varphi_2)$, and so on. Then the operator $\bullet$ is defined as the formal inner product, where the operation between components will be the appropriate in each case, e.g.
\begin{align*}
	\bar{f}\bullet J\bar{\varphi} &= f_1\varphi_2 - f_2\varphi_1 \\
	\bar{V}\bullet J\bar{f} &= V_1(f_2) - V_2(f_1) 
\end{align*}
Some useful properties involving both operators are:
\begin{align}
	\label{o1}
	&\bar{f}\bullet J\bar{f} = 0 \\
	\label{o2}
	&\bar{f}\bullet J\bar{g} + \bar{g}\bullet J\bar{f} =0 \\
	\label{o3}
	&\left(\bar{V}\bullet J\bar{f}\right)' = \bar{V}\bullet J\bar{f}' + \bar{\rho}\bullet J\bar{f}' \\
	\label{o4}
	&\left(\bar{V}\bullet J\bar{f}\right)'' = \bar{V}\bullet J\bar{f}'' + 2\bar{\rho}\bullet J\bar{f}'' +\bar{\rho}'\bullet J\bar{f}' \\
	\label{o5}
	&\left(\bar{f}\bullet J\bar{g}\right)' = \bar{f}'\bullet J\bar{g} + \bar{f}\bullet J\bar{g}'  \\
	\label{o6}
	&h\cdot \left(\bar{V}\bullet J\bar{g}\right) = \bar{V}\bullet J(h\cdot \bar{g})-\bar{V}(h)\bullet J\bar{g}
\end{align}

 A simple calculation leads to
	\begin{align}\label{cc}
			f_{12} &= \bar{V}\bullet J\bar{f} + \bar{\psi}\bullet J\bar{g} + \bar{\varphi}\bullet J\bar{h} \nonumber \\
			g_{12} &= \bar{V}\bullet J\bar{g} + \bar{\varphi}\bullet J\bar{f} + \bar{\alpha}\bullet J\bar{g}  \\
			h_{12} &= \bar{V}\bullet J\bar{h} + \bar{\psi}\bullet J\bar{f} - \bar{\alpha}\bullet J\bar{h} \nonumber 
	\end{align}
where $[V_1,V_2]_{\gamma} =f_{12}T +g_{12}N + h_{12}B$. In obtaining the forthcoming formulas in the rest of the proof we shall use repeatedly formulas \eqref{o1}--\eqref{cc}.	
To prove \textit{(a)} we can check that
	\begin{align}\label{psi12}
		\psi_{12} &= \bar{V}\bullet J\bar{\psi} - \bar{\alpha}\bullet J\bar{\psi}  -G \bar{h}\bullet J\bar{f}
	\end{align}	
Thus, since $V_i \in\mathfrak{X}^{*}_{\mathcal{P}}(\gamma)$ then $\bar{\psi}'-\tau\bar{\psi}=-G\bar{h}$. A simple computation using \eqref{psi12} implies that $\psi_{12}' =\tau\psi_{12}-Gh_{12}$, which means that $[V_1,V_2]_{\gamma}\in\mathfrak{X}^*_{\mathcal{P}}(\gamma)$.
The paragraph \textit{(b)} yields from
	\begin{align}
		[V_1,V_2]_{\gamma} (\tau) &= \bar{V}\bullet J\left(\bar{V}(\tau) \right) = V_1(V_2 (\tau))- V_2(V_1 (\tau)),
	\end{align}
	for what the following relationships have been used: 
	\begin{align}\label{varphi12}
		\varphi_{12} &= \bar{V}\bullet J\bar{\varphi} + \bar{\alpha}\bullet J\bar{\varphi}  -G \bar{g}\bullet J\bar{f} \\
		\label{rho12}
		\rho_{12} &= \bar{V}\bullet J\bar{\rho} \\
		\alpha_{12} &= \bar{V}\bullet J\bar{\alpha}-\bar{\psi}\bullet J\bar{\varphi}-G\bar{h}\bullet J\bar{g}.
	\end{align}	
The linearity will be useful to prove \textit{(d)}, namely, it is sufficient to show \eqref{formula-curvatura} by replacing $U$ for each element of the Frenet frame. We obtain
	\begin{align*}
		D_{[V_1,V_2]_{\gamma}}T+D_{V_2}D_{V_1}T -D_{V_1}D_{V_2}T &= G\left(\bar{f}\bullet J\bar{g}N+\bar{f}\bullet J\bar{h}B\right) \\
		&= G\left(\langle T,V_1\rangle V_2 - \langle T,V_2\rangle V_1\right) \\
		D_{[V_1,V_2]_{\gamma}}N+D_{V_2}D_{V_1}N -D_{V_1}D_{V_2}N &= G\left(\bar{f}\bullet J\bar{h}T+\bar{g}\bullet J\bar{h}N\right) \\
		&= G\left(\langle N,V_1\rangle V_2 - \langle N,V_2\rangle V_1\right) \\
		D_{[V_1,V_2]_{\gamma}}B+D_{V_2}D_{V_1}B -D_{V_1}D_{V_2}B &= G\left(\bar{f}\bullet J\bar{g}T+\bar{h}\bullet J\bar{g}B\right) \\
		&= G\left(\langle B,V_1\rangle V_2 - \langle B,V_2\rangle V_1\right) 
	\end{align*}
Finally, the paragraph \textit{(c)} is immediate by definition, and \textit{(e)} can be easily deduced from \eqref{rho12}.
\end{proof}

\section{Geometric hierarchies of pseudo-null curve flows}\label{sect-hierarchies}
As a further application of the Lie bracket on the set of local vector fields in a pseudo-null curve defined in Section \ref{sect-lie-algebra}, we shall devote the last section to the searching of a hierarchy of infinity  many local symmetries for the pseudo-null vortex filament equation. The algebraic background presented here is a mere adaptation of the one developed in \cite{del_amor_hamiltonian_2014,del_amor_null_2016}. For that reason, those results that do not require any special handling are given without proofs. 

Consider $\Lambda$ the space of arc-length parametrized pseudo-null curves in the Lorentzian space form $M^3_1(G)$. A map $\f:\Lambda \rightarrow \mathcal{C}^\infty(I,\mathbb{R})$ is referred to a scalar field on $\Lambda$ and $\f(\gamma)$ will be also denoted by $\f_\gamma$.
Let $\mathcal{A}$ be the algebra of $\mathcal{P}$-valued scalar fields on $\Lambda$, i.e., if $\f\in \mathcal{A}$, then $\f_\gamma\in \mathcal{P}$ for all $\gamma\in \Lambda$. In this sense, we will also understand the pseudo-torsion scalar field $\ntau:\Lambda \rightarrow \mathcal{C}^\infty(I,\mathbb{R})$  with its obvious meaning.

Similarly, a map $\V:\Lambda \rightarrow \cup_{\gamma\in \Lambda} T_\gamma\Lambda$ is referred to as a vector field on $\Lambda$, and $\V(\gamma)$ will be also denoted by $\V_\gamma$. We shall denote the set of tangent vector fields on $\Lambda$ as $\mathfrak{X}(\Lambda)$, and within we consider the subset $\mathfrak{X}_{\mathcal{A}}(\Lambda)$ of vector fields $\V$ such that $\V_{\gamma}\in \mathfrak{X}_{\mathcal{P}}(\gamma)$. According to \eqref{tangent-gamma}, if we denote $\V=\f\T+\g \N+\h \B$, then
\begin{equation*}
	\mathfrak{X}_{\mathcal{A}}(\Lambda)=\left\{\V\in \mathfrak{X}(\Lambda): \f,\g,\h\in \mathcal{A};\, \h''-2 \ntau  \h'+(\ntau^2-\ntau'+G) \h=0;\, \f'=-\h \right\},
\end{equation*}
where the derivative and anti-derivative operators act on scalar fields as $\f'(\gamma)=\f'_\gamma$ and $\D^{-1}_{s}(\f)(\gamma)=D^{-1}_{s}(\f_\gamma)$ respectively. However, as $\h$ is a $\mathcal{P}$-valuated scalar field, the condition $\h''-2 \ntau  \h'+(\ntau^2-\ntau'+G) \h=0$ is fulfilled if and only if $\h=0$, and therefore we obtain that  
\begin{equation}\label{vector-fields}
	\mathfrak{X}_{\mathcal{A}}(\Lambda)=\left\{\V\in \mathfrak{X}(\Lambda): \f'=0; \,\g\in \mathcal{A};\, \h=0\right\}.
\end{equation}
Thus $\mathfrak{X}_{\mathcal{A}}(\Lambda)$ stands for the set of $\mathcal{A}$-local vector fields locally preserving the arc-length parameter, the pseudo-null character and the hyperplanarity. These vector fields commute with the tangent vector field $\T$, so they will be called evolution vector fields. We also denote by $\bar{\mathfrak{X}}_{\mathcal{A}}(\Lambda)$ and $\mathfrak{X}^*_{\mathcal{A}}(\Lambda)$ the sets of vector fields $\V$ such that $\V_\gamma\in \mathfrak{X}_{\mathcal{P}}(\gamma)$ and $\V_\gamma\in \mathfrak{X}^*_{\mathcal{P}}(\gamma)$, respectively. Hence,
\begin{equation*}
\begin{aligned}
	\bar{\mathfrak{X}}_{\mathcal{A}}(\Lambda)&=\left\{\V=\f\T+\g \N+\h \B: \f,\g,\h\in \mathcal{A}\right\}. \\
	\mathfrak{X}^*_{\mathcal{A}}(\Lambda)&=\left\{\V\in \bar{\mathfrak{X}}_{\mathcal{A}}(\Lambda): \h=0\right\}.
\end{aligned}
\end{equation*}
\begin{remark}
	In what follows, we shall operate with scalar fields and vector fields in the natural way, understanding that the result of the operation is again a scalar field or vector field. For instance, if $\V,\U$ are vector fields on $\Lambda$, then $\left\langle \V,\U \right\rangle$ is a scalar field, where $\left\langle \V,\U \right\rangle(\gamma)=\left\langle \V_\gamma,\U_\gamma \right\rangle$; or $\nabla_{\T} \V$ is again a vector field, where $\nabla_{\T} \V(\gamma)=\nabla_{\T_\gamma}\V_\gamma$ and so on.
\end{remark}
Hence, for $\V\in \mathfrak{X}^*_{\mathcal{A}}(\Lambda)$, the operator 
$D_{\V}:\bar{\mathfrak{X}}_{\mathcal{A}}(\Lambda) \rightarrow \bar{\mathfrak{X}}_{\mathcal{A}}(\Lambda)$ 
is defined as $(D_{\V}\U)(\gamma)=D_{\V_\gamma}\U_\gamma$. The operator $D_{\V}$ can also be described in other words when $\V\in \mathfrak{X}_{\mathcal{A}}(\Lambda)$. Consider $\gamma$ a pseudo-null curve in $\Lambda$,  and suppose that $\V_{\gamma}(s)=\frac{\partial \gamma}{\partial t}(s,0)$, then  
$$\ds (D_{\V}\U)_\gamma(s)=\left.\frac{D}{\partial t}\right|_{t = 0}\U_{\gamma_t}(s).$$
In fact, the tensor derivation $D_{\V}$ is an extension of the Fréchet derivative defined in \eqref{frechet-derivative} for derivations of vector fields on the pseudo-null curves space. In this way, this operator can be easily translated to the context of any other type of curves.

The Lie algebra structure on local vector fields locally preserving the pseudo-null character provided by Theorem \ref{prop-lie-bracket} (along a particular curve) can also be easily extended on the set $\mathfrak{X}^*_{\mathcal{A}}(\Lambda)$.

\begin{theorem}\label{prop-LieBracket2}
	The map  $[\cdot,\cdot]:\mathfrak{X}^*_{\mathcal{A}}(\Lambda)\times \mathfrak{X}^*_{\mathcal{A}}(\Lambda)\rightarrow \mathfrak{X}^*_{\mathcal{A}}(\Lambda)$  given by
	$$[\V,\U](\gamma)=[\V_{\gamma},\U_{\gamma}]_\gamma$$  is a Lie bracket verifying the following
	\begin{enumerate}[(a)]
		\item $[\V,\U](\f)=\V\U(\f)-\U\V(\f)$ for all $\f\in \mathcal{A}$.
		\item $[\cdot,\cdot]$ is closed for elements in $\mathfrak{X}_{\mathcal{A}}(\Lambda)$, i.e., if $\V,\U \in \mathfrak{X}_{\mathcal{A}}(\Lambda)$, then $[\V,\U]\in \mathfrak{X}_{\mathcal{A}}(\Lambda)$.
	\end{enumerate}
	Hence, $[\cdot,\cdot]$ is a Lie bracket, ($\mathfrak{X}^*_{\mathcal{A}}(\Lambda),[,])$ is a Lie algebra and the space of evolution vector fields $(\mathfrak{X}_{\mathcal{A}}(\Lambda),[,])$ is a Lie subalgebra of $\mathfrak{X}^*_{\mathcal{A}}(\Lambda)$.
\end{theorem}

Let us define $\mathop{\rm der}\nolimits (\mathcal{A})$ as the set of derivations on $\mathcal{A}$ defined in the natural way and  $\mathop{\rm der}\nolimits^*(\mathcal{A})$ the Lie subalgebra of all evolution derivations. In this setting, the elements of $\mathop{\rm der}\nolimits^*(\mathcal{A})$ are given by $\Partial_{\p}$, with $\p\in \mathcal{A}$, such that they are defined as usual by $\Partial_{\p}\f(\gamma)=\Partial_{\p_\gamma}\f_\gamma$, for all $\f\in\mathcal{A}$ and $\gamma\in\Lambda$. Each vector field $\V$ on $\Lambda$ can be regarded as a derivation on $\mathcal{A}$, 
acting on the generator $\ntau$ in the following way: $$\V(\ntau)(\gamma)=\V_{\gamma}(\ntau_{\gamma}).$$

\begin{theorem}\label{main-theorem}
	The map $\Phi:\mathfrak{X}_{\mathcal{A}}(\Lambda) \rightarrow \mathop{\rm der}\nolimits^*(\mathcal{A})$ defined by $$\Phi(\V)=\Partial_{\V(\ntau)}=\Partial_{(\mathcal{R}^{2}+G)(\g_{\V})},$$ 
	where $\mathcal{R}$ is defined in \eqref{mburger-recursion}, is a one-to-one homomorphism of Lie algebras. In particular, $\V_{\bm{1}}$ and $\V_{\bm{2}}$ are commuting vector fields with respect to the Lie bracket defined by Proposition \ref{prop-LieBracket2} if and only if their corresponding pseudo-torsion flows $\V_{\bm{1}}(\ntau)$ and $\V_{\bm{2}}(\ntau)$ commute with respect to the usual Lie bracket given by \eqref{scalar-bracket}.
\end{theorem}

\begin{remark}\label{main-remark}
	From Theorem \ref{main-theorem} we have that $\mathop{\rm Im}\nolimits(\Phi)$ is a Lie subalgebra of the algebra   $\mathop{\rm der}\nolimits^*(\mathcal{A})$ of all evolution derivations. Thus, we conclude  that the algebra of evolution vector fields $\mathfrak{X}_{\mathcal{A}}(\Lambda)$ on $\Lambda$ can be regarded as a Lie subalgebra of the evolution derivations. 
\end{remark}
Consider the vector field $\V=\N$. Its flow $\gamma_{t}\in\Lambda$ is governed by the pseudo-null vortex filament equation
\begin{align}
	\label{VF}
	\frac{d}{dt}(\gamma_t)&=\V_{\gamma_{t}}= \N_{\gamma_{t}}.
\end{align}
This induces an evolution equation for the pseudo-torsion $\ntau$ given by 
\begin{align}\label{curve-flow-one}
		\frac{d}{dt}(\ntau_{\gamma_{t}})=\V_{\gamma_{t}}(\ntau_{\gamma_{t}})= \ntau''_{\gamma_{t}}+2 \ntau_{\gamma_{t}}  \ntau'_{\gamma_{t}}.
\end{align}
In other words, if we set $\tau(s,t)=\ntau_{\gamma_{t}}(s)$, equation \eqref{curve-flow-one} becomes the Burgers' equation
\begin{equation}
	\tau_t(s,t)=\tau_{ss}(s,t)+2\tau(s,t)\tau_s(s,t).
\end{equation}
Theorem \ref{main-theorem} will be used below to obtain a recursion operator for the pseudo-null vortex filament equation.

\begin{proposition}\label{prop-recursion}
	The operator $\R$ acting on symmetries as follows 
	\begin{equation}
		\R(\V)= \nabla_{\T}\V, 
	\end{equation} 
	is a recursion operator for the pseudo-null vortex filament equation.	
\end{proposition}
\begin{proof}
	Note that $\R(\T)=\N$ and $\R(\N)=\tau \N$. Moreover, if $\V=\g_{\V} \N$ with $\g_{\V}\in \mathcal{A}_{0}$, then we have that
	\begin{equation*}
		\R(\V)=\nabla_{\T}(\g_{\V}\N)=(\g'_{\V}+\g_{\V}\ntau)\N=\left(D_{s}^{-1}\mathcal{R}D_{s}(\g_{\V})\right)\N.
	\end{equation*}
	Set $\U=\R(\V)$. By definition and making use of the equation \eqref{eq-Vtau} we obtain 	
	\begin{align*}
		\U(\ntau)&=(\mathcal{R}^{2}+G)D_{s}(\g_{\U})=(\mathcal{R}^{2}+G)D_{s}D_{s}^{-1}\mathcal{R}D_{s}(\g_{\V}) \\
		&=\mathcal{R}^{3}D_{s}(\g_{\V})+\mathcal{R} G D_{s}(\g_{\V})=\mathcal{R}(\mathcal{R}^2+G)(D_{s}\g_{\V})  \\
		&=\mathcal{R}(\V(\ntau)),
	\end{align*}
	where $\mathcal{R}$ is given by \eqref{mburger-recursion}. 
	The result can be easily deduced as a consequence of Theorem \ref{main-theorem} and Remark \ref{remark-recursion}.
\end{proof}

According to Proposition \ref{prop-recursion} we proceed to list the first vector fields of the geometric hierarchy:
\begin{align*}
	\V_{\bm{0}} &= \T \\
	\V_{\bm{1}} &= \N = \k \nxi\\
	\V_{\bm{2}} &= \ntau\N = \k' \nxi\\
	\V_{\bm{3}} &= \left(\ntau'+\ntau^{2}\right)\N = \k'' \nxi \\
	\V_{\bm{4}} &= \left(\ntau''+3\ntau\ntau'+\ntau^{3}\right)\N = \k''' \nxi.
\end{align*}
From a physical point of view, Remark \ref{remark-string} means that the vector fields of the geometric hierarchy (except for the first one) generate flows of pseudo-null curves contained in the same (cylindrical) null string and evolving in the direction of its null generator.

Using \eqref{eq5} the above geometric hierarchy induces the Burgers' hierarchy for the pseudo-torsion:
\begin{align*}
	\V_{\bm{0}}(\ntau) &= \ntau'  \\
	\V_{\bm{1}}(\ntau) &= \ntau''+2 \ntau  \ntau' \\
	\V_{\bm{2}}(\ntau) &= \ntau ^{(3)}+3 \ntau  \ntau''+3 \left(\ntau'\right)^2+\left(G+3 \ntau^2\right) \ntau' \\
	\V_{\bm{3}}(\ntau) &= \ntau ^{(4)} + 4\ntau  \ntau ^{(3)} + \left(G+10 \ntau '+6 \ntau ^2\right)\ntau'' + 12 \ntau \left(\ntau'\right)^2+\left(2 G \ntau +4 \ntau^3\right) \ntau' \\
	\V_{\bm{4}}(\ntau) &= \ntau^{(5)} + 5 \ntau \ntau^{(4)} + \left(G+15 \ntau'+ 10 \ntau^2\right) \ntau^{(3)}   \\ 
	&\quad +10 \left(\ntau''\right)^2 + \left(3 G \ntau + 50\ntau  \ntau' + 10\ntau^3 \right) \ntau'' \\
	&\quad +15 \left(\ntau'\right)^3 + \left(3 G+30 \ntau^2\right) \left(\ntau'\right)^2 +\left(3 G \ntau^2 +5\ntau^4 \right) \ntau'
\end{align*}
Likewise, using \eqref{eq-curvature-evolution} we obtain the corresponding pseudo-curvature hierarchy:
\begin{align*}
	\V_{\bm{0}}(\k) &= \k'+\bm{d} \k \\
	\V_{\bm{1}}(\k) &= \k''+G\k+\bm{d} \k \\
	\V_{\bm{2}}(\k) &= \k'''+ G\k'+\bm{d} \k \\
	\V_{\bm{3}}(\k) &= \k^{(4)}+ G\k''+\bm{d} \k \\
	\V_{\bm{4}}(\k) &= \k^{(5)}+ G\k'''+\bm{d} \k
\end{align*}
where $\bm{d}$ verifies that $\bm{d}_{\gamma}$ is constant for each $\gamma\in\Lambda$. By setting $k(s,t)=\k_{\gamma_t}(s)$, the hierarchy of the pseudo-curvature flows can be written as
\begin{align*}
	k_t(s,t)&=k_s(s,t)+d(t)k(s,t) \\
	k_t(s,t)&=k_{ss}(s,t)+G k(s,t)+d(t)k(s,t) \\
	k_t(s,t)&=k_{sss}(s,t)+G k_{s}(s,t)+d(t)k(s,t) \\
	k_t(s,t)&=k_{ssss}(s,t)+G k_{ss}(s,t)+d(t)k(s,t) \\
	k_t(s,t)&=k_{sssss}(s,t)+G k_{sss}(s,t)+d(t)k(s,t) 
\end{align*}
Now, let $D(t)$ be a primitive function of $-d(t)$ and consider the transformation $\tilde{k}(s,t)=e^{D(t)}k(s,t)$.
Then the corresponding evolution equations for $\tilde{k}$ turn out to be the heat hierarchy, i.e.
\begin{align*}
	\tilde{k}_t(s,t)&=\tilde{k}_s(s,t) \\
	\tilde{k}_t(s,t)&=\tilde{k}_{ss}(s,t)+G \tilde{k}(s,t) \\
	\tilde{k}_t(s,t)&=\tilde{k}_{sss}(s,t)+G \tilde{k}_{s}(s,t) \\
	\tilde{k}_t(s,t)&=\tilde{k}_{ssss}(s,t)+G \tilde{k}_{ss}(s,t) \\
	\tilde{k}_t(s,t)&=\tilde{k}_{sssss}(s,t)+G \tilde{k}_{sss}(s,t) 
\end{align*}
Interestingly, for each $t$,  $\tilde{k}(s,t)$ is actually a different pseudo-curvature of $\gamma_t$ with respect to a scaling (with factor $e^{D(t)}$) of the original parallel frame whose associated pseudo-curvature is $k(s,t)$. 

Burgers' equation is the standard prototype of an integrable equation through linearization, thus explicit solutions of Burgers' equation can be calculated from explicit solutions of the linear heat equation through the Hopf-Cole transformation. Furthermore, any solution of the Burgers' equation can be obtained by transforming a solution of the heat equation by means of the Hopf-Cole transformation. In the setting of pseudo-null curves, this fact can be interpreted from a geometrical point of view: they are both (the burgers' equation and the heat equation) evolution equations of two different invariants (the pseudo-torsion and an appropriate pseudo-curvature) of a pseudo-null curve evolving according to the pseudo-null vortex filament equation. Observe also that both the Burgers' hierarchy and the heat hierarchy have been generated from the same geometric hierarchy starting from the vortex filament equation for pseudo-null curves.

\section{Conclusions}
The methods developed in \cite{del_amor_hamiltonian_2014,del_amor_null_2016} to study hierarchies of commuting symmetries for null curve motions flows have been successfully applied for pseudo-null curves. Unlike of the null curves case for which the system was Hamiltonian, the system obtained here is dissipative, but the method works out in a similar way. Essentially, what we have achieved is to lift the integrability properties of Burgers' equation at the pseudo-torsion function level to the pseudo-null vortex filament equation at the curve level (Proposition \ref{prop-recursion}) through the one-to-one homomorphism between their corresponding phase spaces (Theorem \ref{main-theorem}). This lifting construction was the base for the construction of a recursion operator and so the hierarchy of commuting symmetries of the pseudo-null vortex filament equation (Proposition \ref{prop-recursion}). As a physical interpretation we have noted that the vector field flows of the geometric hierarchy can be seen as pseudo-null curves evolving in a null string. 

Additionally, we have interpreted  at the curve level the relationship between the Burgers' equation and the heat equation through the Hopf-Cole transformation. Indeed, both equations can be seen as evolution equations of different invariants ($\tau$ and $k$) coming from the same geometric evolution equation, namely, the pseudo-null vortex filament equation.

\section*{Acknowledgments}
This work has been partially supported by MINECO (Ministerio de Econom\'\i a y Competi\-ti\-vi\-dad) and FEDER project MTM2015-65430-P, and by Fundaci\'on S\'eneca (Regi\'on de Murcia) project 19901/GERM/15,  Spain.


\end{document}